\documentclass[letterpaper,journal]{IEEEtran}

\usepackage{amsmath,amsfonts}
\usepackage{amsthm}
\usepackage{physics}
\usepackage{amssymb}
\usepackage{algorithm}
\usepackage{algpseudocode}
\usepackage{array}
\usepackage{bm}
\usepackage{centernot}
\usepackage{mathtools}
\usepackage{xcolor}
\usepackage{tikz}
\usetikzlibrary{arrows.meta, positioning, shapes.geometric, shadows}
\usepackage{textcomp}
\usepackage{stfloats}
\usepackage{url}
\usepackage{verbatim}
\usepackage{graphicx}
\usepackage{cite}
\usepackage{pdfpages}
\usepackage{hyperref}
\hypersetup{hidelinks}
\usepackage{booktabs}
\usetikzlibrary{calc}
\usetikzlibrary{3d}

\newtheorem{theorem}{Theorem}
\newtheorem{lemma}{Lemma}
\newtheorem{corollary}{Corollary}
\newtheorem{proposition}{Proposition}
\newtheorem{definition}{Definition}
\begin{document}

\title{%
Access Selection for Finite-SNR Modal Recoverability in Sampled-Wave Receivers}
\author{Shaojie Zhang,~\IEEEmembership{Student             Member,~IEEE}
        and Ozgur B. Akan,~\IEEEmembership{Fellow,~IEEE}
        \thanks{The authors are with the Centre for neXt Communications (CXC), 
Department of Engineering, University of Cambridge, CB3 0FA Cambridge, UK 
(e-mail: sz466@cam.ac.uk; oba21@cam.ac.uk).}% <-this % stops a space
\thanks{Ozgur B. Akan is also with the Centre for neXt Communications (CXC), 
Department of Electrical and Electronics Engineering, Ko\c{c} University, 
34450 Istanbul, T\"{u}rkiye (e-mail: akan@ku.edu.tr).}% <-this % stops a space
}


\maketitle

\begin{abstract}
Large-aperture wave receivers can contain a large number of candidate sensor
locations, antenna ports, or measurement blocks, while hardware and processing
constraints allow only a subset to be activated. % [S1]
In this paper, receiver selection is formulated as a finite-SNR modal-recoverability
problem, where the selected subset is required to support reliable recovery in
every direction of a prescribed modal subspace. % [S2]
However, large trace, log-det, or codebook-distance values alone do not ensure
that every prescribed modal direction is resolved. % [S3]
Specifically, the proposed framework consists of three nested recoverability
criteria for individual modal degrees, a joint target subspace, and target modes
in the presence of active non-target modes. % [S4]
The third criterion, the Schur-complement information floor, provides an
exact worst-direction posterior-error interpretation. % [S5]
We further show that stricter recoverability criteria require
at least as many active receiver accesses and derive tests that
identify reliability targets that remain unattainable even when
all available accesses are activated.
Next, we specialize the framework to finite spherical-wave sampling and compare
greedy receiver-selection rules. % [S7]
Numerical results demonstrate that global log-det is generally more
access-efficient at moderate reliability floors, whereas Schur-based selection
is more effective at stringent floors. % [S8]
While this paper is motivated by holographic and XL-MIMO receivers, the
framework can be applied to general sampled wave systems. % [S9]
\end{abstract}


\section{Introduction}

\IEEEPARstart{P}{hysically} large, dense, and software-controlled
apertures have enabled a new class of programmable sampled-wave
receivers for future wireless systems. % [S1]
Holographic MIMO (HMIMO) and extremely large-scale MIMO (XL-MIMO)
architectures can support near-field focusing, integrated sensing and
communications, spatial multiplexing, and programmable electromagnetic
environments~\cite{gong2024hmimo,wei2026eit_hmimo,dong2025debrisense,dong2026martian,dong2026hardware}. % [S2]
For these systems, channel models and estimators must account for
wavefront curvature, aperture size, and geometry-dependent sampling
effects~\cite{cui2022xlmimo,kosasih2023parametric}, while
electromagnetic-information-theoretic formulations characterize the
physical constraints on observable fields~\cite{zhu2024eit}. % [S3]
Therefore, finite-SNR modal recoverability is jointly determined by
the aperture, accessible measurements, sampling geometry, and receiver
noise. % [S4]

In practice, a large-aperture receiver contains a large number of
candidate sensor locations, antenna ports, or measurement blocks, while
hardware and processing constraints allow only a subset to be activated. % [S5]
Therefore, receiver selection for preserving reliable recovery in every
direction of a prescribed modal subspace is crucial. % [S6]
Existing methods, such as log-det, trace/Fisher-information, and
codebook-distance optimization~\cite{krause2008sensor,joshi2009sensor,shamaiah2010greedy,hashemi2021randomized},
can be applied to rank candidate subsets according to aggregate
information. % [S7]
However, optimizing these aggregate objectives alone does not enable
certification of the least-observed modal direction. % [S8]
As a result, a subset can attain a high aggregate score while retaining
a weak or unrecoverable modal direction. % [S9]
The key difference between subset construction and recoverability
certification is that the former selects a candidate subset, whereas
the latter determines its operational adequacy. % [S10]
Given a finite candidate access set, an access budget, a target degree
set, and a reliability threshold, this paper aims to select receiver
accesses and determine whether the required modal components are
recoverable at finite SNR. % [S11]

The proposed formulation is generic and is not restricted to a
particular HMIMO implementation. % [S12]
Spatial degrees of freedom are jointly determined by aperture geometry,
scattering support, and field constraints, rather than by antenna count
alone~\cite{poon2005dof,poon2006scattering}. % [S13]
Wave-channel and electromagnetic formulations characterize field modes
and their coupling strengths~\cite{miller2000waves,migliore2008eit},
while electromagnetic sampling theory describes the associated
degree-of-freedom structure~\cite{pizzo2022nyquist}. % [S14]
Modal representations can be applied to analyze sampled-wave
receivers~\cite{abhayapala2010spherical}. % [S15]
Exact spherical sampling, harmonic transforms, tight frames, and matrix
concentration can be used to construct and analyze receiver
layouts~\cite{driscoll1994sphere,mcewen2011sphere,khalid2014optimal,benedetto2003tight,tropp2012tail}. % [S16]
However, these tools alone do not enable a finite-SNR worst-direction
recoverability test for a selected subset, particularly when active
non-target modes remain coupled to the target subspace. % [S17]
To address this issue, we develop a general sampled scalar-wave model. % [S18]
While this paper is motivated by HMIMO and XL-MIMO, the framework can
be applied to any wave receiver with selectively activated
measurements. % [S19]

The main contributions of this paper are summarized as follows. % [S20]
\begin{itemize}
\item \textit{Problem formulation:}
We formulate receiver selection as a finite-SNR modal recoverability
problem to preserve every required modal direction using a limited
number of activated measurements. % [S21]
Specifically, the proposed framework consists of three increasingly
stringent tests for individual target degrees, the joint target
subspace, and target modes coupled to active non-target modes. % [S22]

\item \textit{Recoverability guarantees:}
We first show that the three tests are nested and that the minimum
required access budget cannot decrease as the test is tightened. % [S23]
Next, we establish an exact worst-direction posterior-error
interpretation for the Schur-complement test and quantify the
information loss caused by target--non-target coupling. % [S24]

\item \textit{Limits of receiver selection:}
We do not treat selector failure and full-set infeasibility as the
same event, but instead distinguish a path-dependent failure from a
limitation of the available measurements. % [S25]
For a monotone library, if the full candidate set fails the required
recoverability test, then no selected subset can satisfy it. % [S26]
However, large trace, log-det, or finite-codebook-distance values
alone do not guarantee recovery of every target modal direction. % [S27]

\item \textit{Design insights:}
To verify the proposed framework, we conduct controlled
spherical-wave experiments and compare degree-wise, target-subspace,
Schur-based, and aggregate selection rules. % [S28]
Numerical results demonstrate that global log-det generally reaches
moderate reliability floors with fewer activated measurements,
whereas Schur-based selection is more effective under stringent
requirements. % [S29]
\end{itemize}

The remainder of this paper is organized as follows. Section~II describes the sampled model. Section~III develops the admission certificates. Section~IV gives the scalar-wave layout laws. Section~V provides the operational guarantees. Section~VI specifies the selection protocol. Numerical results are presented in Section~VII. Section~VIII concludes the paper.

%\begin{figure}[t]
%\centering
%\makebox[\columnwidth][c]%{\includegraphics[width=0.5\textwidth]{figures/modal_recoverability_concept.png}}
%\caption{Certificate-first modal-layer admission workflow.}
%\label{fig:certificate-first-workflow}
%\end{figure}

\section{Sampled Modal Receiver and Recoverability Model}

\subsection{Modal Representation and Receiver Access Library}

We consider a finite modal truncation of order \(L\). The retained index set and its dimension are given by
\[
    \mathcal I_L=\{(\ell,m):0\le \ell\le L,\;-\ell\le m\le \ell\},
    \qquad
    Q=(L+1)^2 .
\]
After truncation, let \(x\) denote the modal coefficient vector of the transmitted or incident field, and let \(C_x\) denote its covariance:
\[
    x\in\mathbb C^Q,
    \qquad
    C_x=\mathbb E[xx^{\mathsf H}]\succeq0 .
\]
Suppose that \(C_x\) is rank deficient. Its compact eigendecomposition is given by
\[
    C_x=U_xP_xU_x^{\mathsf H},
    \qquad
    P_x\succ0,
    \qquad
    r_x=\operatorname{rank}(C_x),
\]
which gives the active-coordinate representation
\[
    x=U_xP_x^{1/2}s,
    \qquad
    \mathbb E[ss^{\mathsf H}]=I_{r_x}.
\]

We consider a receiver access library that consists of the fixed finite set \(\mathcal V\) of physically available accesses before selection. The available accesses can include sensor locations, antenna ports, frequency samples, radii, time gates, and calibrated measurement blocks. The receiver selector selects a subset of the fixed library:
\[
    S\subseteq\mathcal V .
\]
For a selected set \(S\), let \(m_S\) denote the dimension of the stacked receiver output. The corresponding sampled modal matrix is given by
\[
    A_S\in\mathbb C^{m_S\times Q}.
\]
The sampled modal matrix \(A_S\) is determined by receiver geometry, access constraints, and the propagation model, where aperture geometry controls the sampled wavefront structure in field-domain and near-field receivers~\cite{pizzo2022nyquist,cui2022xlmimo}.

\subsection{Sampled Gaussian Model and Recoverability Grams}

For the Gaussian branch, the selected receiver output is modeled as
\[
    y_S=A_Sx+n_S,
    \qquad
    n_S\sim\mathcal{CN}(0,R_S),
    \qquad
    R_S\succ0 ,
\]
where \(R_S\) denotes the selected noise covariance. The covariance \(R_S\) accounts for independent sample noise and correlated receiver noise. By applying receiver-noise whitening, we obtain
\[
    \bar y_S
    =
    R_S^{-1/2}y_S
    =
    R_S^{-1/2}A_Sx+\bar n_S,
    \qquad
    \bar n_S\sim\mathcal{CN}(0,I_{m_S}).
\]
Given the active-coordinate representation of \(x\), the whitened model is given by
\[
    \bar y_S=H(S)s+\bar n_S,
    \qquad
    H(S)=R_S^{-1/2}A_SU_xP_x^{1/2}.
\]

We next define the active source-whitened recoverability Gram as
\[
    K(S)=H(S)^{\mathsf H}H(S).
\]
Let \(F(S)\) and \(G(S)\) denote the unweighted normal Gram and the prior-weighted full-coordinate recoverability Gram, respectively, where
\[
    F(S)=A_S^{\mathsf H}R_S^{-1}A_S,
    \qquad
    G(S)=C_x^{1/2}F(S)C_x^{1/2}.
\]
The matrix \(F(S)\) measures sampled information in physical modal coordinates, while \(G(S)\) includes the modal covariance. When \(C_x\) is singular, \(K(S)\) provides the equivalent recoverability Gram in active coordinates.

\subsection{Modal Subspaces and the Additive-Gram Convention}

For each modal degree \(\ell\), let \(V_\ell\) and \(d_\ell\) denote the associated subspace and its dimension, respectively. % [S1]
\[
    V_\ell=\operatorname{span}\{e_{\ell,m}:m=-\ell,\ldots,\ell\},
    \qquad
    d_\ell=2\ell+1,
\]
Let \(\Pi_\ell\) denote the orthogonal projector onto \(V_\ell\). % [S2]
We assume that the target degree set
\(\mathcal T\subseteq\{0,\ldots,L\}\) is specified before the receiver-selection stage. % [S3]
To evaluate finite-SNR certificates, we restrict \(G(S)\), or \(F(S)\) for an unweighted physical certificate, to the relevant modal subspaces. % [S4]

For access measurements with independent noise blocks, the prior-weighted Gram is given by the following additive form. % [S5]
\[
    G(S)=\sum_{e\in S}B_e,
    \qquad
    B_e=C_x^{1/2}A_e^{\mathsf H}R_e^{-1}A_eC_x^{1/2}\succeq0.
\]
Although this additive form applies to independent noise blocks, we do not impose it on the general model. % [S6]
For correlated receiver noise, shared analog front ends, or joint processing, the full selected covariance \(R_S\) is used to construct \(G(S)\). % [S7]
Lemma~\ref{lem:monotone_growth} provides the sufficient conditions required by the subsequent monotonicity arguments. % [S8]

\section{Finite-SNR Modal Recoverability Criteria}

In this section, we first construct a degree-wise information test for one modal degree. % [S1]
Next, we introduce a target-subspace test for the joint activation of the prescribed modal layers. % [S1]
At last, we develop a Schur-complement test for the setting in which the non-target modes remain Gaussian nuisance coefficients. % [S1]

\subsection{Degree-Wise Information Floors}

Consider a selected access subset \(S\subseteq\mathcal V\). % [S2]
Let \(G(S)\) and \(F(S)\) denote the prior-weighted recoverability Gram and the unweighted normal Gram introduced in Section~II\@, respectively. % [S3]
For modal degree \(\ell\), we define the prior-weighted degree information floor as follows. % [S9]
\begin{equation}
    \mu_\ell(S)
    =
    \lambda_{\min}
    \left(
        \left.
        \Pi_\ell G(S)\Pi_\ell
        \right|_{V_\ell}
    \right)
    =
    \inf_{\substack{u\in V_\ell\\ \|u\|_2=1}}
    u^{\mathsf H}G(S)u.
    \label{eq:mu_degree_certificate}
\end{equation}
Here, \(\Pi_\ell\) denotes the orthogonal projector onto \(V_\ell\). Likewise, the unweighted modal information floor in physical coordinates is given by the following expression. % [S5]
\begin{equation}
    \nu_\ell(S)
    =
    \lambda_{\min}
    \left(
        \left.
        \Pi_\ell F(S)\Pi_\ell
        \right|_{V_\ell}
    \right)
    =
    \inf_{\substack{u\in V_\ell\\ \|u\|_2=1}}
    u^{\mathsf H}F(S)u.
    \label{eq:nu_degree_certificate}
\end{equation}
Here, \(\mu_\ell(S)\) represents the prior-weighted information floor, whereas \(\nu_\ell(S)\) represents the unweighted physical-coordinate floor. % [S6]
In this paper, the admission rules below focus on \(\mu_\ell(S)\), whereas \(\nu_\ell(S)\) is invoked only when an unweighted test is stated explicitly. % [S7]

Given a fixed threshold \(\tau>0\), we use the following degree-\(\ell\) admission condition. % [S8]
\begin{equation}
    \mu_\ell(S)\ge \tau.
    \label{eq:degree_cert_condition}
\end{equation}
Consider a prescribed target degree set
\(T \subseteq \{0,\ldots,L\}\), and define
\begin{equation}
\mu_{\mathrm{deg}}(S)
=
\min_{\ell\in T}\mu_{\ell}(S).
\label{eq:degree_target_floor}
\end{equation}
The subset \(S\) satisfies degree-wise target admission when
\(\mu_{\mathrm{deg}}(S)\geq \tau\). % [S22]

\noindent\emph{Joint target-subspace admission:}
We next define the target modal subspace and its associated projector. % [S4]
\[
V_T=\bigoplus_{\ell\in T}V_\ell, \qquad \Pi_T=\sum_{\ell\in T}\Pi_\ell.
\]
Based on these definitions, we use the following target-subspace information floor. % [S19]
\[
\mu_T(S)=\lambda_{\min}\!\left(\Pi_TG(S)\Pi_T\big|_{V_T}\right).
\]
Given \(\tau\), the subset \(S\) jointly admits the target modal stack when \(\mu_T(S)\ge\tau\). % [S8]
However, this joint condition can fail even if every degree in \(T\) meets its degree-wise requirement. % [S14]

\noindent\emph{Admission with Gaussian nuisance modes:}
We further retain the modal coefficients outside the target subspace. % [S16]
Let \(V_c=V_T^\perp\) denote the non-target subspace within the retained modal coordinates. % [S3]
With respect to \(V_T\oplus V_c\), the prior-weighted Gram is given by the following block form. % [S12]
\[
G(S)=\begin{bmatrix}G_{TT}(S)&G_{Tc}(S)\\G_{cT}(S)&G_{cc}(S)\end{bmatrix}.
\]
According to the Schur complement, we use the following nuisance-conditioned target information matrix. % [S10]
\[
I_{T|c}(S)=G_{TT}(S)-G_{Tc}(S)\big(I+G_{cc}(S)\big)^{-1}G_{cT}(S).
\]
We next define the associated nuisance-conditioned information floor. % [S4]
\[
\mu_{T|c}(S)=\lambda_{\min}\big(I_{T|c}(S)\big).
\]
Given \(\tau\), the subset \(S\) satisfies nuisance-robust target admission when \(\mu_{T|c}(S)\ge\tau\). % [S8]
In this case, the non-target coefficients are modeled as active Gaussian nuisance variables rather than deleted coordinates. % [S25]

\begin{theorem}[Nested modal-admission certificates]
\label{thm:multitier-admission}
Assume that the source-whitened Gaussian model holds. % [S15]
Let the retained coefficient vector be partitioned as \(s=(s_T,s_c)\), where \(s_T\in V_T\) and \(s_c\in V_c\). % [S3]
For a selected access set \(S\), marginalizing \(s_c\) gives the following posterior covariance of \(s_T\). % [S12]
\[
\Sigma_{T|y}(S)=\big(I+I_{T|c}(S)\big)^{-1}.
\]
As a result, the nuisance-robust information-floor condition and the worst posterior-variance condition are equivalent. % [S11]
\[
\mu_{T|c}(S)\ge\tau
\quad\Longleftrightarrow\quad
\lambda_{\max}\big(\Sigma_{T|y}(S)\big)\le (1+\tau)^{-1}.
\]
Moreover, the three information floors obey the following ordering. % [S16]
\[
0\le \mu_{T|c}(S)\le \mu_T(S)\le \min_{\ell\in T}\mu_\ell(S).
\]
Therefore, nuisance-robust admission guarantees joint target-subspace admission, and joint target-subspace admission guarantees degree-wise admission. % [S12]
However, the two reverse implications are not guaranteed in general. % [S14]
\end{theorem}

\begin{IEEEproof}
See Appendix~\ref{app:proof_multitier_admission}.
\end{IEEEproof}

\begin{corollary}[Intermodal coupling penalties and nested access budgets]
\label{cor:coupling-tax}
Let \(G_{\ell j}(S)=\Pi_\ell G(S)\Pi_j|_{V_j\to V_\ell}\) denote the information block from \(V_j\) to \(V_\ell\). % [S3]
Applying the block eigenvalue bound gives the following lower bound. % [S18]
\[
\mu_T(S)\ge
\min_{\ell\in T}\left[\mu_\ell(S)-\sum_{j\in T,\,j\ne\ell}\|G_{\ell j}(S)\|_2\right].
\]
Applying the spectral-norm bound to the Schur correction gives the following lower bound. % [S18]
\[
\mu_{T|c}(S)\ge \mu_T(S)-\|G_{Tc}(S)(I+G_{cc}(S))^{-1}G_{cT}(S)\|_2.
\]
To compare the access requirements, we define the minimum feasible budgets as follows. % [S9]
\begin{align*}
K_{\deg}(\tau)
&=\min\{|S|:\min_{\ell\in T}\mu_\ell(S)\ge\tau\},\\
K_T(\tau)
&=\min\{|S|:\mu_T(S)\ge\tau\},\\
K_{T|c}(\tau)
&=\min\{|S|:\mu_{T|c}(S)\ge\tau\}.
\end{align*}
If a corresponding feasible set is empty, its minimum budget is set to \(\infty\). % [S20]
Therefore, the three access budgets satisfy the following nesting relation. % [S12]
\[
K_{\deg}(\tau)\le K_T(\tau)\le K_{T|c}(\tau).
\]
\end{corollary}

\begin{IEEEproof}
See Appendix~\ref{app:proof_coupling_tax}.
\end{IEEEproof}

As a result, the three criteria provide a degree-restricted E-optimal information floor and a fixed-library finite-SNR feasibility gate for modal communication layers. % [S17]

Due to a non-degree-isotropic source covariance \(C_x\), the values of \(\mu_\ell(S)\) and \(\nu_\ell(S)\) can be different. % [S21]
Consider degree-isotropic covariance blocks with no cross-degree covariance. % [S2]
\begin{equation}
    C_x\big|_{V_\ell}=P_\ell I_{2\ell+1}.
\end{equation}
Under this covariance model, the relation between the two floors is given by the following expression. % [S12]
\begin{equation}
    \mu_\ell(S)=P_\ell\nu_\ell(S).
    \label{eq:mu_nu_relation}
\end{equation}
If cross-degree covariance is present, \(\mu_\ell(S)\) continues to incorporate the prior weighting. % [S20]
Therefore, physical degree-amplitude recovery uses \(\nu_\ell(S)\) or the corresponding Schur certificate in physical coordinates. % [S12]

\subsection{Monotonicity of Fixed Access Libraries}

The full-set infeasibility result is based on monotone information growth. % [S1]

\begin{definition}[Monotone finite candidate measurement set]
A finite candidate measurement set \(\mathcal V\) is referred to as \(G\)-monotone under the condition below. % [S2]
For every \(S\subseteq S'\subseteq\mathcal V\), the defining matrix ordering takes the following form. % [S3]
\begin{equation}
    G(S)\preceq G(S').
\end{equation}
Likewise, the \(F\)-monotonicity condition is given by \(F(S)\preceq F(S')\) for every \(S\subseteq S'\subseteq\mathcal V\). % [S4]
\end{definition}

The two standard sufficient cases are summarized below. % [S5]

\begin{lemma}[Sufficient conditions for monotone information growth]
\label{lem:monotone_growth}
Each of the following conditions is sufficient for the fixed
library to be \(G\)-monotone.

\emph{1) Additive independent access blocks:}
Suppose
\begin{equation}
    G(S)=\sum_{e\in S}B_e,
    \qquad
    B_e\succeq0.
\end{equation}
Then \(S\subseteq S'\) implies
\[
    G(S)\preceq G(S').
\]

\emph{2) Nested marginals of a fixed jointly Gaussian observation:}
Suppose the full-library observation satisfies
\begin{equation}
    y_{\mathcal V}=A_{\mathcal V}x+n_{\mathcal V},
    \qquad
    n_{\mathcal V}\sim\mathcal{CN}(0,R_{\mathcal V}),
    \qquad
    R_{\mathcal V}\succ0,
\end{equation}
where \(R_{\mathcal V}\) is independent of \(x\). If the
observations indexed by \(S\subseteq S'\subseteq\mathcal V\)
are nested marginals of this full observation, then
\begin{equation}
    F(S)\preceq F(S'),
    \qquad
    G(S)\preceq G(S').
\end{equation}
\end{lemma}

\begin{IEEEproof}
See Appendix~\ref{app:proof_monotonicity}.
\end{IEEEproof}

\begin{corollary}[Full-library infeasibility]
\label{cor:full_library_infeasibility}
Suppose that the finite candidate measurement set \(\mathcal V\) is
\(G\)-monotone. Then, for every \(S\subseteq\mathcal V\),
\[
    \mu_\ell(S)\leq\mu_\ell(\mathcal V),
    \qquad
    \mu_{\mathcal T}(S)\leq\mu_{\mathcal T}(\mathcal V),
    \qquad
    \mu_{\mathcal T|c}(S)\leq\mu_{\mathcal T|c}(\mathcal V).
\]
Consequently, if the full library fails any prescribed
degree-wise, target-subspace, or nuisance-robust threshold,
no subset selector over \(\mathcal V\) can satisfy the
corresponding certificate.

If \(\mathcal V\) is \(F\)-monotone, then
\[
    \nu_\ell(S)\leq\nu_\ell(\mathcal V),
    \qquad S\subseteq\mathcal V,
\]
and the analogous full-library infeasibility statement holds
for the unweighted physical certificate.
\end{corollary}

\begin{IEEEproof}
See Appendix~\ref{app:proof_full_library}.
\end{IEEEproof}

\subsection{Separation Between Aggregate Objectives and Admission Certificates}
\label{subsec:aggregate-separation}

The multi-tier information floors introduced above are used to determine finite-SNR feasibility for modal admission. % [B1]
Common aggregate objectives can be applied to rank candidate access subsets. % [B2]
However, their objective values alone do not provide a degree-wise modal-admission certificate. % [B3]
Let \(G_\ell(S)\) denote the restriction of the prior-weighted Gram to degree \(\ell\). % [B4]
\begin{equation}
G_\ell(S)=\Pi_\ell G(S)\Pi_\ell|_{V_\ell}.
\end{equation}
For a finite source-whitened codeword-difference set
\(\mathcal D\subset V_\ell\setminus\{0\}\), we define the following three aggregate objectives. % [B5]
\begin{align}
\Phi_{\rm tr}(S) &= \operatorname{tr}G_\ell(S),\\
\Phi_{\log}(S) &= \log_2\det(I+G_\ell(S)),\\
\Phi_{\mathcal D}(S) &= \min_{\Delta\in\mathcal D}\Delta^H G_\ell(S)\Delta .
\end{align}

For physical-coordinate codebooks, let \(F_\ell(S)=\Pi_\ell F(S)\Pi_\ell|_{V_\ell}\). % [B6]
If \(G_\ell(S)=P_\ell F_\ell(S)\), then the prior-weighted and physical codebook-distance objectives differ only by the positive scalar \(P_\ell\). % [B7]

\begin{theorem}[Aggregate-objective separation]
\label{thm:aggregate-certification-separation}
Consider a modal degree \(\ell\ge 1\), dimension \(d_\ell=2\ell+1\), degree power \(P_\ell>0\), access budget \(K\ge d_\ell\), threshold \(\tau>0\), finite nonzero difference set \(\mathcal D\subset V_\ell\setminus\{0\}\), and objective gap \(R>0\). % [B8]
Under these parameters, there exist a distributed subset \(S_{\rm sp}\) and an ideal repeated-access clustered limit \(S_{\rm cl}^{\rm rep}\) with \(|S_{\rm cl}^{\rm rep}|=|S_{\rm sp}|=K\). % [B9]
The associated modal-floor relations are given below. % [B10]
\begin{equation}
\mu_\ell(S_{\rm sp})\ge \tau,
\qquad
\mu_\ell(S_{\rm cl}^{\rm rep})=0.
\end{equation}
The corresponding aggregate-objective gaps are given below. % [B11]
\begin{align}
\Phi_{\rm tr}(S_{\rm cl}^{\rm rep}) &> \Phi_{\rm tr}(S_{\rm sp})+R,\\
\Phi_{\log}(S_{\rm cl}^{\rm rep}) &> \Phi_{\log}(S_{\rm sp})+R,\\
\Phi_{\mathcal D}(S_{\rm cl}^{\rm rep}) &> \Phi_{\mathcal D}(S_{\rm sp})+R.
\end{align}
Given any \(\epsilon>0\), the angular cap can be chosen sufficiently small for the physical distinct-access construction. % [B12]
To obtain a finite scalar-wave access library, the repeated clustered accesses are replaced by \(K\) distinct accesses within this cap. % [B13]
As a result, the finite scalar-wave subset \(S_{\rm cl}\) satisfies \(|S_{\rm cl}|=K\) and \(0\le \mu_\ell(S_{\rm cl})<\epsilon\), while retaining the three aggregate-objective gaps. % [B14]
By applying the continuity of the eigenvalue and aggregate-objective maps, the three strict inequalities persist under this distinct-access replacement. % [B15]
Accordingly, repeated identical angular responses specify the rank-deficient limiting construction, while distinct accesses realize the construction in a physical finite library. % [B16]
\end{theorem}

\begin{IEEEproof}
See Appendix~\ref{app:proof-aggregate-separation}.
\end{IEEEproof}

As a result, the theorem establishes the following structural non-implication. % [B17] 
A high aggregate score is not, by itself, a finite-SNR modal-recoverability certificate.
The same construction can also be applied to a larger truncation by assigning inactive or identical responses to the non-target degrees. % [B18]
As a result, full-space trace and log-det criteria can exhibit the same separation. % [B19]
Aggregate objectives can still be applied to subset ranking and post-admission optimization. % [B20]
However, their values alone do not establish finite-SNR modal feasibility. % [B21]

\section{Recoverability Conditions for Finite Wave-Receiver Layouts}
\label{sec:wave_receiver_certificates}

This section applies the modal-admission framework to finite sampled scalar-wave matrices. % [S1]
For finite scalar-wave layouts, a weighted access library provides an exact certificate, while a separable access structure yields a closed-form design rule. % [S2]

\subsection{Exact Weighted Scalar-Wave Certificate}

For a sampled scalar-wave access \(e\in\mathcal V\), the received sample is given by the following model. % [S3]
\begin{equation}
    y_e
    =
    \sum_{\ell=0}^{L}
    \sum_{m=-\ell}^{\ell}
    a_{\ell}^{\rm wav}(e)Y_\ell^m(\Omega_e)x_{\ell m}
    +n_e,
    \label{eq:weighted_wave_sample}
\end{equation}
Here, \(\Omega_e\) denotes the receiver direction, \(a_\ell^{\rm wav}(e)\) denotes the degree-\(\ell\) transfer coefficient associated with the selected frequency, radius, and gain, and \(n_e\sim\mathcal{CN}(0,\sigma_e^2)\) is the receiver noise. % [S4]
Let \(y_\ell(\Omega_e)\) denote the degree-\(\ell\) angular evaluation vector. % [S5]
\begin{equation}
    y_\ell(\Omega_e)
    =
    \big[
    Y_\ell^{-\ell}(\Omega_e),\ldots,Y_\ell^{\ell}(\Omega_e)
    \big]^{\mathsf T}.
\end{equation}

\begin{proposition}[Exact certificate for weighted finite scalar-wave layouts]
\label{thm:weighted_wave_certificate}
The analysis assumes the scalar-wave observation model in \eqref{eq:weighted_wave_sample}, independent Gaussian receiver noise, and the degree-isotropic covariance \(C_x|_{V_\ell}=P_\ell I_{2\ell+1}\). % [S6]
Let \(w_{\ell,e}\) denote the degree-\(\ell\) access weight. % [S7]
\begin{equation}
    w_{\ell,e}
    =
    \frac{|a_\ell^{\rm wav}(e)|^2}{\sigma_e^2}.
\end{equation}
Under these assumptions, the degree-\(\ell\) finite-SNR recoverability margin is given by the following expression. % [S8]
\begin{equation}
    \mu_\ell(S)
    =
    P_\ell
    \lambda_{\min}
    \left(
    \sum_{e\in S}
    w_{\ell,e}
    y_\ell(\Omega_e)y_\ell(\Omega_e)^{\mathsf H}
    \right).
    \label{eq:weighted_wave_certificate}
\end{equation}
For \(P_\ell>0\), degree-\(\ell\) certification by \(S\) at threshold \(\tau\) is equivalent to the following eigenvalue condition. % [S9]
\begin{equation}
    \lambda_{\min}
    \left(
    \sum_{e\in S}
    w_{\ell,e}
    y_\ell(\Omega_e)y_\ell(\Omega_e)^{\mathsf H}
    \right)
    \ge
    \frac{\tau}{P_\ell}.
\end{equation}
\end{proposition}

\begin{IEEEproof}
See Appendix~\ref{app:proof_weighted_wave_certificate}.
\end{IEEEproof}

Here, \(w_{\ell,e}\) incorporates the degree-dependent transfer gain and receiver-noise penalty, while \(y_\ell(\Omega_e)\) characterizes the angular sampling response. % [S10]

\begin{corollary}[Closed-form law under separable transfer and angular sampling]
\label{cor:separable_wave_law}
Consider a selected access library with a product, grouped, or equal-weight structure, under which the weighted angular Gram is separable. % [S11]
The resulting factorization is given by the following expression. % [S12]
\begin{equation}
    \sum_{e\in S}
    w_{\ell,e}
    y_\ell(\Omega_e)y_\ell(\Omega_e)^{\mathsf H}
    =
    \frac{B_{M,\ell}(S)}{\sigma_M^2}\Gamma_\ell,
\end{equation}
Here, \(B_{M,\ell}(S)\ge0\) denotes the degree-\(\ell\) transfer-energy sum, and \(\Gamma_\ell\) denotes the degree-restricted angular sampling Gram. % [S13]
As a result, the degree-\(\ell\) information margin can be written as follows. % [S14]
\begin{equation}
    \mu_\ell(S)
    =
    \frac{P_\ell}{\sigma_M^2}
    B_{M,\ell}(S)\lambda_{\min}(\Gamma_\ell).
    \label{eq:finite_layout_margin}
\end{equation}
Suppose that \(P_\ell\lambda_{\min}(\Gamma_\ell)>0\). % [S15]
Degree \(\ell\) is certified when the following transfer-energy condition is satisfied. % [S16]
\begin{equation}
    B_{M,\ell}(S)
    \ge
    \frac{\sigma_M^2\tau}{P_\ell\lambda_{\min}(\Gamma_\ell)}.
    \label{eq:exact_access_energy_threshold}
\end{equation}
\end{corollary}

\begin{IEEEproof}
See Appendix~\ref{app:proof_separable_wave_law}.
\end{IEEEproof}

\subsection{Residual-Interval Certificates}

For a general finite angular layout, we consider the following residual representation. % [S1]
\begin{equation}
    \Gamma_\ell
    =
    \frac{N_{\rm sens}}{4\pi}I_{2\ell+1}+E_\ell,
    \qquad
    \|E_\ell\|_2\le \epsilon_\ell,
    \label{eq:residual_sampling_model}
\end{equation}
Here, \(N_{\rm sens}\) denotes the number of angular samples, and \(\epsilon_\ell\) denotes an upper bound on the spectral norm of the residual. % [S2]

We next define the nominal degree-\(\ell\) information level as follows. % [S3]
\begin{equation}
    \rho_\ell^{(0)}(S)
    =
    \frac{N_{\rm sens}P_\ell}
    {4\pi\sigma_M^2}
    B_{M,\ell}(S).
\end{equation}
Furthermore, we define the corresponding residual radius as follows. % [S4]
\begin{equation}
    \delta_\ell(S)
    =
    \frac{P_\ell}{\sigma_M^2}
    B_{M,\ell}(S)\epsilon_\ell.
\end{equation}
As a result, the residual certificate interval can be written as follows. % [S5]
\begin{equation}
    I_\ell(S)
    =
    \big[
    \rho_\ell^{(0)}(S)-\delta_\ell(S),
    \rho_\ell^{(0)}(S)+\delta_\ell(S)
    \big].
    \label{eq:residual_interval}
\end{equation}

\begin{proposition}[Residual-interval certification labels]
\label{prop:residual_labels}
Under \eqref{eq:residual_sampling_model}, \(I_\ell(S)\) contains every Rayleigh quotient of the degree-\(\ell\) information matrix evaluated at a unit vector. % [S6]
The resulting sufficient certification and rejection labels are given below. % [S7]
\begin{align}
    \rho_\ell^{(0)}(S)-\delta_\ell(S)\ge\tau
    &\Rightarrow
    S\ \text{certifies degree}\ \ell,\\
    \rho_\ell^{(0)}(S)+\delta_\ell(S)<\tau
    &\Rightarrow
    S\ \text{does not certify degree}\ \ell.
\end{align}
If the rejection condition is satisfied for \(S=\mathcal V\) and the access library is monotone, then no selector restricted to subsets of \(\mathcal V\) can certify degree \(\ell\). % [S8]
\end{proposition}

\begin{IEEEproof}
See Appendix~\ref{app:proof_residual_labels}.
\end{IEEEproof}

When neither implication applies, the residual interval alone does not enable a certification decision. % [S9]
To resolve this ambiguity, an exact information-floor evaluation or a tighter spectral-residual bound is required. % [S10]

\subsection{Angular Sampling Efficiency and the Tight-Frame Bound}

Let \(y_\ell(\Omega_i)\) denote the degree-\(\ell\) evaluation vector associated with the receiver direction \(\Omega_i\in\mathbb S^2\). % [S1]
\begin{equation}
    y_\ell(\Omega_i)
    =
    \big[
    Y_\ell^{-\ell}(\Omega_i),
    \ldots,
    Y_\ell^\ell(\Omega_i)
    \big]^{\mathsf T}
    \in\mathbb C^{2\ell+1}.
\end{equation}
We next define the degree-\(\ell\) angular Gram of the layout \(\Omega_{1:N_{\rm sens}}\) as follows. % [S2]
\begin{equation}
    \Gamma_\ell(\Omega_{1:N_{\rm sens}})
    =
    \sum_{i=1}^{N_{\rm sens}}
    y_\ell(\Omega_i)y_\ell(\Omega_i)^{\mathsf H}.
\end{equation}
The angular efficiency of this layout is denoted by \(\alpha_\ell\) and defined as follows. % [S3]
\begin{equation}
    \alpha_\ell
    =
    \frac{4\pi}{N_{\rm sens}}
    \lambda_{\min}(\Gamma_\ell).
    \label{eq:angular_efficiency}
\end{equation}

\begin{lemma}[Angular-efficiency bound and tight-frame equality]
\label{thm:angular_efficiency}
For any angular layout, the admissible range of \(\alpha_\ell\) is given by the following expression. % [S4]
\begin{equation}
    0\le \alpha_\ell\le 1.
\end{equation}
Under Corollary~\ref{cor:separable_wave_law}, the degree-\(\ell\) recoverability margin can be written as follows. % [S5]
\begin{equation}
    \mu_\ell(S)
    =
    \frac{P_\ell}{\sigma_M^2}
    B_{M,\ell}(S)
    \frac{N_{\rm sens}}{4\pi}
    \alpha_\ell.
    \label{eq:angular_efficiency_margin}
\end{equation}
The equality \(\alpha_\ell=1\) is equivalent to the following degree-\(\ell\) tight-frame relation. % [S6]
\begin{equation}
    \Gamma_\ell
    =
    \frac{N_{\rm sens}}{4\pi}I_{2\ell+1}.
\end{equation}
Consequently, for \(\alpha_\ell>0\), the exact transfer-energy requirement can be written as follows. % [S7]
\begin{equation}
    B_{M,\ell}(S)
    \ge
    \frac{4\pi\sigma_M^2\tau}
    {P_\ell N_{\rm sens}\alpha_\ell}.
    \label{eq:angular_efficiency_threshold}
\end{equation}
\end{lemma}

\begin{IEEEproof}
See Appendix~\ref{app:proof_angular_efficiency}.
\end{IEEEproof}

The scalar \(\alpha_\ell\) indicates the receiver-distribution efficiency at modal degree \(\ell\). % [S8]
The certified margins of layouts with the same sensor count and transfer energy can be different from each other. % [S9]
Tight-frame layouts are therefore directly associated with the equality case of finite-SNR modal certification~\cite{benedetto2003tight}, while spherical-harmonic sampling provides well-conditioned reference designs~\cite{driscoll1994sphere,mcewen2011sphere,khalid2014optimal}. % [S10]

\begin{theorem}[Scalar-wave cardinality and transfer-energy converse]
\label{thm:access-energy-converse}
Consider the scalar-wave model in Proposition~\ref{thm:weighted_wave_certificate} and a modal degree \(\ell\) with dimension \(d_\ell=2\ell+1\). % [S11]
If \(S\) certifies degree \(\ell\) at threshold \(\tau>0\), then the following two necessary conditions hold. % [S12]
\[
|S|\ge d_\ell
\]
and
\[
\sum_{e\in S} w_{\ell,e}\ge \frac{4\pi\tau}{P_\ell}.
\]
If \(w_{\ell,e}\le w_{\max,\ell}\) for every feasible access, then each certified subset satisfies the following cardinality bound. % [S13]
\[
|S|\ge \max\left\{d_\ell,\left\lceil\frac{4\pi\tau}{P_\ell w_{\max,\ell}}\right\rceil\right\}.
\]
For simultaneous target-subspace admission, a positive value of either \(\mu_T(S)\) or \(\mu_{T|c}(S)\) further requires \(|S|\ge d_T\), where \(d_T=\sum_{\ell\in T}(2\ell+1)\). % [S14]
\end{theorem}

\begin{IEEEproof}
See Appendix~\ref{app:proof_access_energy_converse}.
\end{IEEEproof}

\begin{proposition}[High-probability guarantee for isotropic random layouts]
\label{thm:random_layout}
Consider \(N_{\rm sens}\) receiver directions sampled independently and uniformly on \(\mathbb S^2\). % [S15]
For \(0<\varepsilon<1\), the lower-tail probability is bounded as follows. % [S16]
\begin{equation}
    \Pr\!\left\{
    \lambda_{\min}(\Gamma_\ell)
    \le
    (1-\varepsilon)\frac{N_{\rm sens}}{4\pi}
    \right\}
    \le
    (2\ell+1)
    \exp\!\left(
    -\frac{N_{\rm sens}\varepsilon^2}
    {2(2\ell+1)}
    \right).
    \label{eq:random_layout_chernoff}
\end{equation}
Consequently, on the complementary event, the degree-\(\ell\) information margin satisfies the following lower bound. % [S17]
\begin{equation}
    \mu_\ell(S)
    \ge
    \frac{P_\ell}{\sigma_M^2}
    B_{M,\ell}(S)
    \frac{N_{\rm sens}}{4\pi}
    (1-\varepsilon).
\end{equation}
\end{proposition}

\begin{IEEEproof}
See Appendix~\ref{app:proof_random_layout}.
\end{IEEEproof}

Proposition~\ref{thm:random_layout} provides a matrix-Chernoff benchmark showing that isotropic random layouts approach the tight-frame scaling law as the number of angular samples increases~\cite{tropp2012tail}. % [S18]

\section{Recoverability Guarantees and Access Limits}
\label{sec:operational}

In this section, we first describe the finite-SNR guarantees induced by degree-wise admission. % [S01]
Next, we present simultaneous target-layer operation under target-subspace admission. % [S02]
At last, we describe target-layer operation under Schur admission when the non-target modes remain active as nuisance variables. % [S03]

\begin{proposition}[Operational implications of modal certificates]
\label{thm:operational_metrics}
Suppose that source-whitened Gaussian signaling is adopted. % [S1]
Let \(G_{\ell\ell}(S)\) denote the degree-\(\ell\) block of the prior-weighted Gram associated with the selected access set \(S\). % [S2]
\begin{equation}
G_{\ell\ell}(S)
\triangleq
\left.\Pi_\ell G(S)\Pi_\ell\right|_{V_\ell}.
\end{equation}
The argument \(S\) is suppressed in the remainder of this proposition for notation simplification. % [S3]
The three effective information triples are given by % [S4]
\begin{equation}
\label{eq:effective_information_triples}
(J,d,\mu)=
\begin{cases}
(G_{\ell\ell},d_\ell,\mu_\ell), & \text{degree-wise},\\
(G_{TT},d_T,\mu_T), & \text{target subspace},\\
(I_{T|c},d_T,\mu_{T|c}), & \text{Schur-conditioned}.
\end{cases}
\end{equation}
If \(\mu\ge\tau\), i.e., the minimum eigenvalue of \(J\) is no smaller than \(\tau\), then the following log-det lower bound holds. % [S5]
\begin{equation}
\log_2\det(I_d+J)
\ge
d\log_2(1+\tau),
\label{eq:operational_logdet}
\end{equation}
The same condition ensures that the effective posterior covariance does not exceed \((1+\tau)^{-1}I_d\). % [S6]
\begin{equation}
(I_d+J)^{-1}
\preceq
(1+\tau)^{-1}I_d.
\label{eq:operational_covariance}
\end{equation}
In particular, \((I_d+J)^{-1}\) indicates the degree-restricted posterior covariance, the target-only posterior covariance, and the nuisance-marginalized target posterior covariance for the three cases, respectively. % [S7]
Inequality~\eqref{eq:operational_covariance} ensures that the largest source-whitened posterior variance does not exceed \((1+\tau)^{-1}\). % [S8]
The same inequality ensures that the worst-direction EVM amplitude does not exceed \((1+\tau)^{-1/2}\). % [S9]

Given two codewords whose source-whitened difference \(\Delta s\) belongs to \(V_\ell\), their squared distance and Gaussian binary ML error can be bounded as follows. % [S10]
\begin{equation}
d^2(S)
\ge
\tau\|\Delta s\|_2^2,
\qquad
P(i\rightarrow j)
\le
\frac{1}{2}
\exp\left(
-\frac{\tau\|\Delta s\|_2^2}{4}
\right).
\label{eq:operational_pairwise}
\end{equation}
\end{proposition}

\begin{IEEEproof}
See Appendix~\ref{app:proof_operational_metrics}.
\end{IEEEproof}

The threshold \(\tau\) can be selected from operational requirements.
For an admission tier of dimension \(d\), a log-det target \(R\) requires
\(\tau\ge 2^{R/d}-1\), and an EVM target
\(\operatorname{EVM}_{\max}\) requires
\(\tau\ge \operatorname{EVM}_{\max}^{-2}-1\).
If all relevant source-whitened codeword differences satisfy
\(\|\Delta s\|_2^2\ge d_{\min}^2\), then
\(P(i\to j)\le P_e^\star\) is guaranteed by
\[
    \tau
    \ge
    \frac{4}{d_{\min}^2}
    \ln\!\left(\frac{1}{2P_e^\star}\right).
\]
If multiple requirements are imposed, the largest resulting lower bound is used.

Given an access library \(\mathcal V\), a budget \(K\), a target set \(T\),
and an operational threshold \(\tau\), the three full-library certificate
values can be evaluated by the receiver as follows. % [S1]
\[
\min_{\ell\in T}\mu_\ell(\mathcal V),\qquad
\mu_T(\mathcal V),\qquad
\mu_{T|c}(\mathcal V).
\]
If the required full-library value is below \(\tau\), i.e., the full
library does not satisfy the required admission tier, then no subset
selector over the same monotone library can satisfy that tier. % [S2]
Otherwise, the access-selection process is iterated until the relevant
certificate is satisfied or the budget \(K\) is exhausted. % [S3]
Degree-wise admission is used to support single-degree operation. % [S4]
Target-subspace admission is used to support simultaneous target-layer
activation. % [S5]
Schur admission is used to support target operation when non-target
modes remain active as nuisance variables. % [S6]
After admission, downstream resource-allocation criteria, such as
log-det, Fisher information, codebook distance, or geometric coverage,
can be applied to allocate the remaining access budget. % [S7]

\section{Receiver-Selection Rules and Baselines}
\label{sec:selection_protocol}

The construction of all receiver-selection paths is carried out in
source- and noise-whitened coordinates. % [S01]
For illustration, we consider the white-prior
branch, where the physical and prior-weighted recoverability
Grams coincide. % [S02]
Under independent scalar accesses, the common recoverability
Gram can be given by % [S03]
\begin{equation}
G(S)=\sum_{e\in S}g_e g_e^{\rm H}.
\label{eq:num_additive_gram}
\end{equation}
The three secondary information criteria are given by % [S05]
\begin{align}
J_{\rm deg}(S)
&=
\sum_{\ell\in T}
\log_2\det\!\left(
I+\Pi_\ell G(S)\Pi_\ell\big|_{V_\ell}
\right),
\label{eq:num_degree_secondary}\\
J_T(S)
&=
\log_2\det\!\left(I+G_{TT}(S)\right),
\label{eq:num_target_secondary}\\
J_{T|c}(S)
&=
\log_2\det\!\left(I+I_{T|c}(S)\right).
\label{eq:num_schur_secondary}
\end{align}

For each \(a\in\{{\rm deg},T,T|c\}\), let \(e_a^\star\)
denote the access selected by the following strict
lexicographic rule. % [S06]
\begin{equation}
e_a^\star
=
\mathop{\arg\max\nolimits^{\rm lex}}_{e\in\mathcal V\setminus S}
\left(
\mu_a(S\cup\{e\}),
J_a(S\cup\{e\})
\right).
\label{eq:num_lexicographic_selector}
\end{equation}
In~\eqref{eq:num_lexicographic_selector},
\(\mu_a\) is the primary coordinate and \(J_a\) breaks
certificate ties within numerical tolerance; any remaining tie
is resolved by the smallest candidate index.

Once the requested certificate reaches its prescribed
threshold, the remaining budget is allocated according to the
global log-det marginal. % [S13]
By applying the matrix determinant lemma to the rank-one
update \(G(S\cup\{e\})=G(S)+g_e g_e^{\rm H}\), we can obtain
the following marginal expression. % [S14]
\begin{equation}
\begin{aligned}
\Delta_{\rm log}(e\mid S)
&=
\log_2\det\!\left(I+G(S\cup\{e\})\right)
-\log_2\det\!\left(I+G(S)\right)\\
&=
\log_2\!\left(
1+
g_e^{\rm H}
\left(I+G(S)\right)^{-1}
g_e
\right),\\
e_{\rm log}^\star
&=
\arg\max_{e\in\mathcal V\setminus S}
\Delta_{\rm log}(e\mid S).
\end{aligned}
\label{eq:num_post_admission_logdet}
\end{equation}
For the global log-det baseline,
\eqref{eq:num_post_admission_logdet} is applied from the
empty set. % [S15]

Random-prefix references use fixed-seed permutations of the same access library and are evaluated at each integer budget.

Consider the structural aggregate-objective separation
diagnostic, where \(\ell_0\) denotes the tested modal degree
and \(d_{\ell_0}=2\ell_0+1\) denotes its dimension. % [S18]
For our numerical study, we set \(\ell_0=3\). % [S19]
The finite codebook is given by the following deterministic
\(d_{\ell_0}\)-dimensional DFT basis. % [S20]
\begin{equation}
c_p[q]
=
\frac{1}{\sqrt{d_{\ell_0}}}
\exp\!\left(
\frac{{\rm j}2\pi pq}{d_{\ell_0}}
\right),
\qquad
p,q=0,\ldots,d_{\ell_0}-1,
\label{eq:num_dft_codebook}
\end{equation}
The DFT basis is embedded in \(V_{\ell_0}\), and the
corresponding difference set can be given by % [S21]
\begin{equation}
\mathcal D_{\ell_0}
=
\left\{
c_p-c_{p'}:
0\leq p<p'<d_{\ell_0}
\right\}.
\label{eq:num_codeword_difference_set}
\end{equation}
The performance quantities considered in this diagnostic are
the trace, degree-restricted log-det, and finite-codebook
distance. % [S22]
However, these quantities are not used to construct additional
selector trajectories. % [S23]
The performance metrics considered for every deterministic
or randomly generated subset are
\(\mu_{\rm deg}(S)\), \(\mu_T(S)\), and
\(\mu_{T|c}(S)\). % [S24]
These three certificates are evaluated regardless of the
objective used to construct the subset. % [S25]

\section{Numerical Evaluation and Discussion}
\label{sec:numerical_results}

The numerical experiments for the proposed admission framework are
carried out in this section. % [S01]
The three claims evaluated in the experiments are summarized as follows:
the admission tiers are distinct, aggregate information objectives cannot
replace a worst-direction certificate, and the relative efficiency of greedy
receiver-selection paths varies with the prescribed finite-SNR floor. % [S02]
For illustration convenience, we consider controlled dimensionless
experiments. % [S03]
Calibration to a particular hardware platform is not considered in this
numerical study. % [S04]

\subsection{Simulation Configuration}
\label{subsec:num_setup}

For our numerical experiments, we set the normalized operating point to
\(kr=5\) and the retained order to
\(L=\lceil kr\rceil+1=6\), including one guard degree beyond the usual
truncation rule~\cite{abhayapala2010spherical}. % [S05]
The resulting retained model contains \(Q=49\) modal coefficients. % [S06]
Let \(T=\{2,3,4\}\) and \(d_T=21\) denote the target band and its
dimension, respectively. % [S07]
The retained nuisance complement contains degrees \(0\)--\(1\) and
\(5\)--\(6\). % [S08]

The remaining simulation parameters are summarized in
Table~\ref{tab:num_setup}. % [S09]
The candidate access library consists of equal numbers of deterministic
Fibonacci directions in a spherical cap and over the full sphere. % [S10]
As a result, the clustered and distributed library components have the
same cardinality, and no library-size bias is introduced. % [S11]
The cap half-angles are configured to be consistent with the order of
angular extents reported in clustered indoor propagation
studies~\cite{poon2005dof}. % [S12]
However, these angles are not interpreted as a particular measured RMS
angular spread. % [S13]

\begin{table}[t]
\centering
\caption{Simulation parameters and configuration.}
\label{tab:num_setup}
\setlength{\tabcolsep}{4pt}
\begin{tabular}{l|l}
\hline
Parameter & Value \\
\hline
Operating point & $kr=5$ \\
Retained model & $L=6$, $Q=49$ \\
Target band & $T=\{2,3,4\}$, $d_T=21$ \\
Candidate library & $N_{\rm cap}=N_{\rm dist}=98$ \\
Cap half-angle & $10^\circ,15^\circ,20^\circ,25^\circ,30^\circ$ \\
Access contrast & $0,5,10,15,20$ dB \\
Certificate floor & $\tau\in\{0.5,1,2\}$ \\
Displayed budgets & $K=0,\ldots,98$ \\
Posterior trials & $N_{\rm post}=50000$ \\
Random permutations & $N_{\rm rand}=2500$ \\
\hline
\end{tabular}
\end{table}

Each candidate access is normalized by its target-subspace gain.
Distributed accesses are assigned unit target gain, whereas cap
accesses are assigned a \(\Delta_{\rm acc}\)-dB target-gain advantage:
\[
\|\Pi_T g_e\|_2^2=\rho_e,\qquad
\rho_e=
\begin{cases}
10^{\Delta_{\rm acc}/10}, & e \text{ in the cap},\\
1, & e \text{ in the distributed set}.
\end{cases}
\]

For each selector \(a\), let \(S_a(K)\) denote the first \(K\) accesses
on its deterministic path. % [S19]
We next define the nuisance-robust first-crossing budget of selector
\(a\) as follows. % [S20]
\begin{equation}
K_{T|c}^{(a)}(\tau)
=
\min\left\{
K\le |\mathcal V|:
\mu_{T|c}\!\left(S_a(K)\right)\ge\tau
\right\}.
\label{eq:num_selector_crossing_budget}
\end{equation}
For comparison purposes, we first verify that the full access library
reaches the requested certificate floor. % [S21]
If the full library does not reach the requested floor, the configuration
is labeled infeasible and is not assigned to either selector. % [S22]

\subsection{Core Certificate Validation}
\label{subsec:num_core_evidence}

To evaluate the claims most directly related to the theory,
Fig.~\ref{fig:num_core_evidence} presents the core numerical checks. % [S23]
Panels~(a)--(c) use the representative configuration, where
\(\theta_{\rm cap}=25^\circ\) and
\(\Delta_{\rm acc}=10\) dB. % [S24]

Panel~(a) shows the three certificate levels at every integer budget,
where the observed ordering is % [S25]
\begin{equation}
\mu_{T|c}(S)\le \mu_T(S)\le \mu_{\rm deg}(S).
\label{eq:num_observed_nesting}
\end{equation}
This ordering confirms that degree-wise admission can be reached before
target-subspace admission, and target-subspace admission can be reached
before nuisance-robust admission. % [S26]

To evaluate the posterior interpretation of the Schur certificate,
panel~(b) presents the worst-direction target EVM. % [S27]
The predicted target EVM is given by the following expression. % [S28]
\begin{equation}
{\rm EVM}_{\rm th}(K)
=
\left[
1+\mu_{T|c}\!\left(S(K)\right)
\right]^{-1/2}.
\label{eq:num_theoretical_evm}
\end{equation}
We evaluate the theoretical posterior-EVM curve at all displayed
budgets, i.e., \(K=0,\ldots,98\). % [S1]
To verify the theoretical prediction, we conduct \(5\times10^4\)
nuisance-active Gaussian Monte Carlo trials at \(25\) representative
budgets, including those in the neighborhoods of the certificate
crossings. % [S2]
The empirical Monte Carlo results validate the theoretical curve
within the reported uncertainty. % [S3]

Panel~(c) plots the nuisance-robust admission probabilities obtained from
every prefix of \(2500\) fixed-seed random permutations. % [S32]
These admission probabilities are used to determine how often an
unstructured access prefix satisfies the Schur certificate, rather than
to measure the performance of a tuned selector. % [S34]
Pointwise binomial confidence bands are shown, where the bands describe
the uncertainty of the random-prefix admission probabilities. % [S35]
For $\tau = 1, 2$, it can be seen that none of the displayed random prefixes
reaches the requested floor. % [S36]

Panel~(d) evaluates the aggregate-objective separation result at
\(\ell=3\) and \(K=2\ell+1=7\). % [S37]
For this numerical diagnostic, we use a \(121\)-point logarithmic sweep
normalized by
\(\rho_\star=\max\{\rho_{\rm tr},\rho_{\log\det},\rho_{\rm cd}\}\). % [S38]
Here, \(\rho_{\rm tr}\), \(\rho_{\log\det}\), and \(\rho_{\rm cd}\)
denote the clustered-to-distributed crossing contrasts for trace,
degree-restricted log-det, and DFT-codebook distance, respectively. % [S39]
When \(\rho/\rho_\star\ge1\), the clustered construction exceeds the
distributed reference in all three aggregate metrics, while its modal
minimum-eigenvalue floor remains zero. % [S40]
This result confirms that aggregate-objective improvement does not
establish modal admission. % [S41]

\begin{figure*}[t]
\centering
\includegraphics[width=0.99\textwidth]
{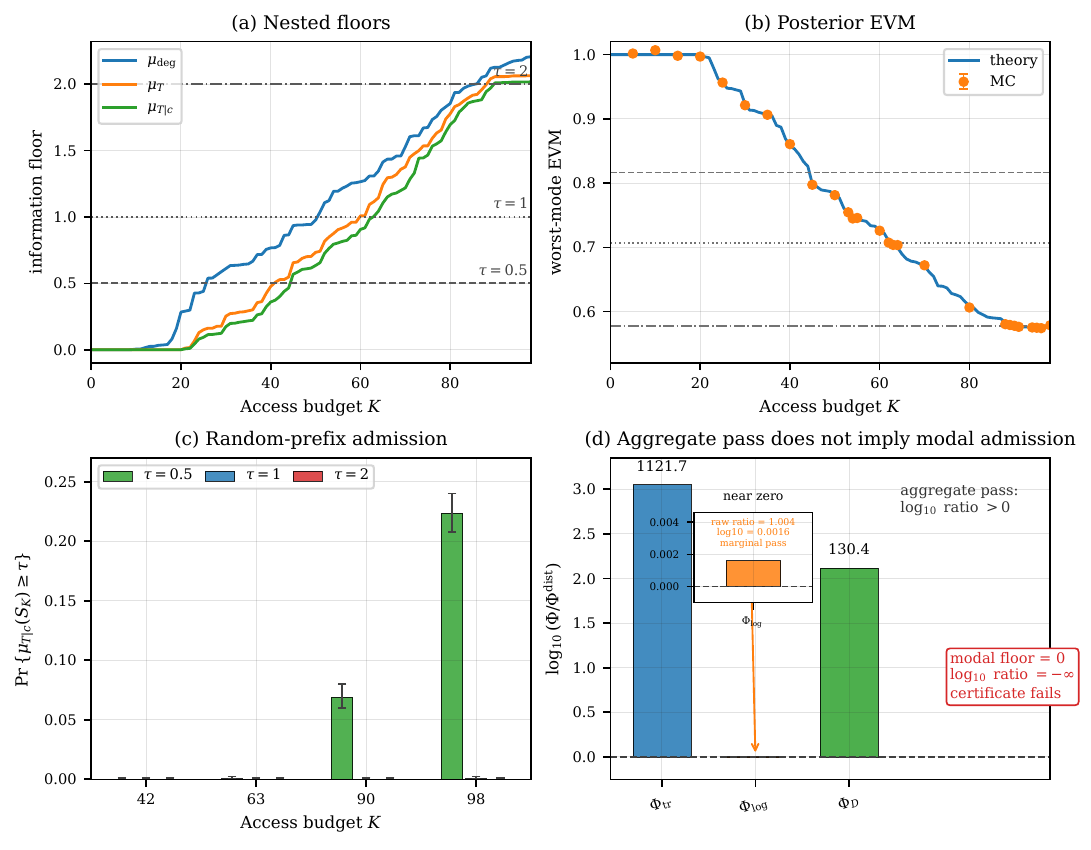}
\caption{Core finite-SNR admission evidence, where panel~(a) shows the
nested degree-wise, target-subspace, and Schur-conditioned floors,
panel~(b) compares the Schur-predicted and empirical worst-direction EVM,
panel~(c) reports random-prefix nuisance-robust admission probabilities,
and panel~(d) shows that aggregate metrics can improve while the modal
minimum-eigenvalue certificate remains unsatisfied.} % [S42]
\label{fig:num_core_evidence}
\end{figure*}

\subsection{Threshold-Dependent First-Crossing Performance}
\label{subsec:num_selector_sensitivity}

To compare the first-crossing budgets of Schur-first and global log-det,
we define \(\Delta K\) as follows. % [S43]
\begin{equation}
\Delta K
=
K_{T|c}^{({\rm Schur})}
-
K_{T|c}^{({\rm logdet})}.
\label{eq:num_crossing_difference}
\end{equation}
In particular, \(\Delta K<0\), \(\Delta K>0\), and
\(\Delta K=0\) indicate an advantage for Schur-first, an advantage for
global log-det, and a tie, respectively. % [S44]
Figure~\ref{fig:num_selector_sensitivity} plots \(\Delta K\) over the
complete \(5\times5\) cap-angle--contrast grid, where each panel
corresponds to one certificate floor. % [S45]
A hollow marker indicates that at least one crossing occurs beyond the
displayed budget \(K_{\max}=98\), while the corresponding deterministic
path is continued to obtain the crossing value. % [S46]

In general, global log-det reaches the two moderate certificate floors
with fewer accesses. % [S47]
For \(\tau=0.5\), the numbers of configurations favoring Schur-first,
yielding a tie, and favoring global log-det are \(1\), \(3\), and \(21\),
respectively. % [S48]
For \(\tau=1\), the corresponding numbers are \(2\), \(1\), and \(22\),
respectively. % [S49]
In contrast, for \(\tau=2\), the corresponding numbers are \(16\), \(3\),
and \(6\), respectively. % [S50]

The grid-wide median of \(\Delta K\) changes from \(+6\) at both
\(\tau=0.5\) and \(\tau=1\) to \(-4\) at \(\tau=2\). % [S51]
For the representative configuration, global log-det reaches
\(\tau=1\) at \(K=54\), compared with \(K=63\) for Schur-first, whereas
Schur-first reaches \(\tau=2\) at \(K=90\), compared with \(K=95\) for
global log-det. % [S52]
The numerical results show that first-crossing efficiency depends on the
prescribed certificate floor. % [S53]
Global log-det is generally more access-efficient at moderate floors,
while direct maximization of the Schur certificate becomes more
effective under a stringent worst-direction requirement. % [S54]
Although trajectory dominance at larger budgets may be relevant, we do
not focus on it in this first-crossing comparison. % [S55] 

\begin{figure*}[t]
\centering
\includegraphics[width=0.99\textwidth]
{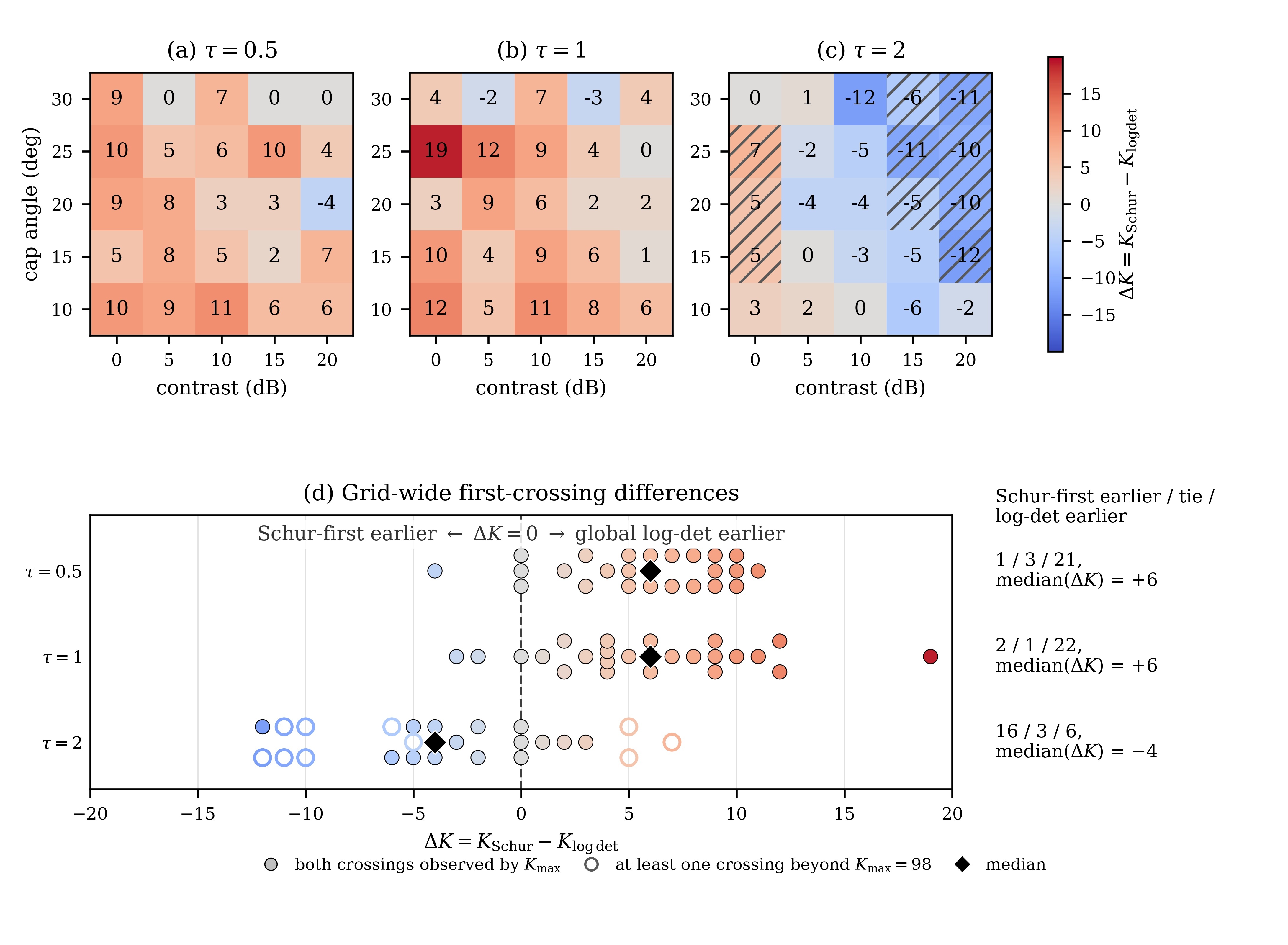}
\caption{Threshold-dependent first-crossing performance, where
panels~(a)--(c) show
\(\Delta K=K_{\rm Schur}-K_{\log\det}\) for
\(\tau=0.5\), \(1\), and \(2\), and panel~(d) summarizes the
grid-wide distributions with black diamonds indicating their medians.} % [S56]
\label{fig:num_selector_sensitivity}
\end{figure*}

\section{Conclusion}

In conclusion, the key distinction between receiver selection and modal feasibility is that
the two procedures address different design tasks. % [S1]
A selection objective is used to determine the order in which candidate
measurements are activated. % [S2]
However, the selection objective alone does not enable the conclusion that
the prescribed modal subspace is recoverable. % [S3]
To this end, modal operation is gated by a worst-direction finite-SNR
requirement. % [S4]
Aggregate objectives can then be applied only to construct or refine a
candidate subset. % [S5]

If the full access library fails the required certificate, i.e., its
certificate value is below the prescribed floor, then the failure represents
a system-level limitation and no subset selector over the same monotone
library can satisfy that tier. % [S6]
In contrast, when the full access library is feasible, the remaining problem
focuses on access-efficient selection, where the preferred rule depends on
the required reliability level. % [S7]
Although universal selector optimality and calibrated HMIMO/XL-MIMO
performance are important, we do not focus on them in the normalized
spherical-wave study, instead, we use the study to illustrate these design
principles. % [S8]

% ===========================
% Appendices
% ===========================

\appendices
\small

\section{Proofs for Finite-SNR Degree Certification}
\label{app:certification_proofs}

\subsection{Proof of the Monotonicity Lemma}
\label{app:proof_monotonicity}
\begin{IEEEproof}
First, consider the case with independent access blocks. Then
\begin{equation}
    G(S')-G(S)
    =
    \sum_{e\in S'\setminus S}B_e
    \succeq0.
\end{equation}

For the correlated Gaussian case, set \(S'=S\cup\mathcal D\), where \(\mathcal D=S'\setminus S\). Partition the selected observation into \((y_S,y_{\mathcal D})\). The Gaussian likelihood can be factored as
\begin{equation}
    p(y_{S'}\mid x)
    =
    p(y_S\mid x)\,
    p(y_{\mathcal D}\mid y_S,x).
\end{equation}
The conditional covariance of \(y_{\mathcal D}\) given \(y_S\) is the Schur complement
\begin{equation}
    R_{\mathcal D|S}
    =
    R_{\mathcal D\mathcal D}
    -
    R_{\mathcal D S}R_{SS}^{-1}R_{S\mathcal D}
    \succ0,
\end{equation}
and the corresponding conditional sensing matrix is
\begin{equation}
    A_{\mathcal D|S}
    =
    A_{\mathcal D}
    -
    R_{\mathcal D S}R_{SS}^{-1}A_S.
\end{equation}
The resulting information increment is
\begin{equation}
    F(S')-F(S)
    =
    A_{\mathcal D|S}^{\mathsf H}
    R_{\mathcal D|S}^{-1}
    A_{\mathcal D|S}
    \succeq0.
\end{equation}
Premultiplication and postmultiplication by \(C_x^{1/2}\) give
\begin{equation}
    G(S')-G(S)
    =
    C_x^{1/2}
    \bigl(F(S')-F(S)\bigr)
    C_x^{1/2}
    \succeq0.
\end{equation}
\end{IEEEproof}

\subsection{Proof of Corollary~\ref{cor:full_library_infeasibility}}
\label{app:proof_full_library}

\begin{IEEEproof}
 By Lemma~\ref{lem:monotone_growth}, \(G(S)\preceq G(V)\) for every
\(S\subseteq V\). Hence, the Rayleigh--Ritz characterization gives
\[
\mu_\ell(S)\le \mu_\ell(V),
\qquad
\mu_T(S)\le \mu_T(V).
\]
Moreover, Loewner monotonicity of inversion yields
\[
\big[(I+G(S))^{-1}\big]_{TT}
\succeq
\big[(I+G(V))^{-1}\big]_{TT}.
\]
Using
\[
\big[(I+G(\cdot))^{-1}\big]_{TT}
=
\big(I+I_{T|c}(\cdot)\big)^{-1},
\]
and inverting the positive-definite target blocks gives
\[
I_{T|c}(S)\preceq I_{T|c}(V),
\]
and therefore
\[
\mu_{T|c}(S)\le \mu_{T|c}(V).
\]
Thus, if the full-library value of any certificate is below
\(\tau\), no subset can satisfy that certificate. The
\(F\)-monotone statement follows identically.
\end{IEEEproof}

\subsection{Proof of Theorem~\ref{thm:multitier-admission}}
\label{app:proof_multitier_admission}

Write the source-whitened observation as \(\bar y=H_Ts_T+H_cs_c+\bar n\), where \(s_T\) and \(s_c\) are independent \(\mathcal{CN}(0,I)\) variables and \(\bar n\sim\mathcal{CN}(0,I)\). The full posterior precision of \((s_T,s_c)\) equals \(I+G(S)\). The \(T\)-block of the posterior covariance is obtained by inverting the Schur complement of the \(c\)-block in \(I+G(S)\), namely
\[
\Sigma_{T|y}(S)=\left(I+G_{TT}-G_{Tc}(I+G_{cc})^{-1}G_{cT}\right)^{-1}.
\]
This identity gives the posterior formula and its equivalence to largest posterior target variance at most \((1+\tau)^{-1}\). Because \(G_{Tc}(I+G_{cc})^{-1}G_{cT}\succeq0\), \(\mu_{T|c}\le \mu_T\). The inequality \(\mu_T\le\mu_\ell\) for every \(\ell\in T\) follows from Cauchy interlacing, and nonnegativity follows from Schur-complement positivity. \(\square\)

\subsection{Proof of Corollary~\ref{cor:coupling-tax}}
\label{app:proof_coupling_tax}

For \(u=\oplus_{\ell\in T}u_\ell\in V_T\), bound the diagonal block terms of the Rayleigh quotient from below by \(\mu_\ell\|u_\ell\|^2\) and the cross terms by the operator-norm triangle inequality. The block-Gershgorin bound follows. Applying Weyl's inequality to \(I_{T|c}=G_{TT}-G_{Tc}(I+G_{cc})^{-1}G_{cT}\) gives the Schur bound. Budget nesting is obtained from the feasible-set inclusions in Theorem~\ref{thm:multitier-admission}. \(\square\)

\subsection{Proof of Theorem~\ref{thm:aggregate-certification-separation}}
\label{app:proof-aggregate-separation}

It is sufficient to prove the repeated identical-response limit. Distinct physical accesses are then obtained through a small-cap continuity argument.

Denote by \(y_\ell(\Omega)\) the degree-\(\ell\) spherical-harmonic evaluation vector. Choose \(K\) directions \(\Omega_1,\ldots,\Omega_K\) for which
\begin{equation}
\Gamma_{\rm sp}=\sum_{i=1}^K y_\ell(\Omega_i)y_\ell(\Omega_i)^H\succ0.
\end{equation}
Such a set exists because \(K\ge d_\ell\). Choose \(w_{\rm sp}>0\) such that
\begin{equation}
P_\ell w_{\rm sp}\lambda_{\min}(\Gamma_{\rm sp})\ge \tau .
\end{equation}
Then degree \(\ell\) is certified by \(S_{\rm sp}\).

Next, select a direction \(\Omega_0\) such that
\begin{equation}
\Delta^H y_\ell(\Omega_0)\ne0,
\qquad \forall\Delta\in\mathcal D .
\end{equation}
For each nonzero \(\Delta\), \(\Delta^H y_\ell(\Omega)\) is a nonzero degree-\(\ell\) spherical harmonic. The zero set of this harmonic has surface measure zero, and finitely many such zero sets cannot cover the sphere. Set
\begin{equation}
\gamma_{\mathcal D}=\min_{\Delta\in\mathcal D}|\Delta^H y_\ell(\Omega_0)|^2>0 .
\end{equation}
Construct \(S_{\rm cl}\) with \(K\) accesses having the same angular response \(y_\ell(\Omega_0)\) and common weight \(w_{\rm cl}\). The degree-restricted Gram is
\begin{equation}
G_\ell(S_{\rm cl})=P_\ell K w_{\rm cl} y_\ell(\Omega_0)y_\ell(\Omega_0)^H .
\end{equation}
This rank-one matrix satisfies
\begin{equation}
\mu_\ell(S_{\rm cl})=\lambda_{\min}(G_\ell(S_{\rm cl}))=0 .
\end{equation}
However,
\begin{align}
\Phi_{\rm tr}(S_{\rm cl})
&=P_\ell K w_{\rm cl}\|y_\ell(\Omega_0)\|_2^2,\\
\Phi_{\log}(S_{\rm cl})
&=\log_2\left(1+P_\ell K w_{\rm cl}\|y_\ell(\Omega_0)\|_2^2\right),\\
\Phi_{\mathcal D}(S_{\rm cl})
&=P_\ell K w_{\rm cl}\gamma_{\mathcal D}.
\end{align}
As \(w_{\rm cl}\to\infty\), all three aggregate quantities diverge while the distributed-subset quantities remain fixed. Select \(w_{\rm cl}\) so that all three gaps exceed \(R\), while \(\mu_\ell(S_{\rm cl})=0\).

For distinct accesses, select \(w_{\rm cl}\) with strict margin larger than \(R\). When a cap around \(\Omega_0\) shrinks, the finite angular Gram and aggregate objectives converge to the repeated-access limit. Hence, for a sufficiently small cap, the smallest eigenvalue is below any prescribed \(\epsilon>0\), and the aggregate inequalities remain valid. \(\square\)

\section{Proofs for Wave-Receiver Finite-Layout Certificates}
\label{app:receiver_design_proofs}

\subsection{Proof of Proposition~\ref{thm:weighted_wave_certificate}}
\label{app:proof_weighted_wave_certificate}

Stacking \eqref{eq:weighted_wave_sample} over \(e\in S\) yields the selected sampling matrix \(A_S\). Under independent noise,
\(G(S)=\sum_{e\in S}C_x^{1/2}A_e^{\mathsf H}\sigma_e^{-2}A_eC_x^{1/2}\).
Restricting to \(V_\ell\) and using \(C_x^{1/2}|_{V_\ell}=\sqrt{P_\ell}I\) yield the compressed Gram
\[
\Pi_\ell G(S)\Pi_\ell|_{V_\ell}
=P_\ell\sum_{e\in S}
\frac{|a_\ell^{\rm wav}(e)|^2}{\sigma_e^2}
y_\ell(\Omega_e)y_\ell(\Omega_e)^{\mathsf H}.
\]
Taking the smallest eigenvalue yields \eqref{eq:weighted_wave_certificate}, and division by \(P_\ell>0\) yields the certification condition.

\subsection{Proof of Corollary~\ref{cor:separable_wave_law}}
\label{app:proof_separable_wave_law}

Substituting the separable Gram into Proposition~\ref{thm:weighted_wave_certificate} factors out the nonnegative scalar \(B_{M,\ell}(S)/\sigma_M^2\), which yields \eqref{eq:finite_layout_margin}. Solving \(\mu_\ell(S)\ge\tau\) yields \eqref{eq:exact_access_energy_threshold}.

\subsection{Proof of Proposition~\ref{prop:residual_labels}}
\label{app:proof_residual_labels}

Substitution yields \(u^{\mathsf H}G(S)u=\rho_\ell^{(0)}(S)+(P_\ell/\sigma_M^2)B_{M,\ell}(S)u^{\mathsf H}E_\ell u\). Because \(|u^{\mathsf H}E_\ell u|\le\epsilon_\ell\), every unit Rayleigh quotient lies in \(I_\ell(S)\). The certificate labels and full-set infeasibility statement then follow.

\subsection{Proof of Lemma~\ref{thm:angular_efficiency}}
\label{app:proof_angular_efficiency}

By the addition theorem, \(\|y_\ell(\Omega)\|_2^2=(2\ell+1)/(4\pi)\), so \(\operatorname{tr}(\Gamma_\ell)=N_{\rm sens}(2\ell+1)/(4\pi)\) and the average eigenvalue is \(N_{\rm sens}/(4\pi)\). Therefore, \(0\le\alpha_\ell\le1\), with equality if and only if \(\Gamma_\ell=(N_{\rm sens}/4\pi)I_{2\ell+1}\). Substitution into Corollary~\ref{cor:separable_wave_law} yields \eqref{eq:angular_efficiency_margin}; solving yields \eqref{eq:angular_efficiency_threshold}.

\subsection{Proof of Theorem~\ref{thm:access-energy-converse}}
\label{app:proof_access_energy_converse}
The degree-\(\ell\) weighted angular Gram consists of \(|S|\) rank-one positive semidefinite matrices. To obtain a positive smallest eigenvalue in dimension \(d_\ell\), one needs \(|S|\ge d_\ell\). Define \(M_\ell(S)=\sum_{e\in S}w_{\ell,e}y_\ell(\Omega_e)y_\ell(\Omega_e)^H\). Since \(\|y_\ell(\Omega)\|_2^2=d_\ell/(4\pi)\),
\[
\lambda_{\min}\big(M_\ell(S)\big)
\le \frac{1}{d_\ell}\operatorname{tr}M_\ell(S)
=\frac{1}{4\pi}\sum_{e\in S}w_{\ell,e}.
\]
Certification requires \(P_\ell\lambda_{\min}\ge\tau\), which yields the transfer-energy bound. The bounded-access lower bound follows directly. For target admission, each scalar access contributes at most one rank, and therefore a positive floor on \(d_T\) dimensions requires \(|S|\ge d_T\). The rank of the Schur information \(I_{T|c}=H_T^H(I+H_cH_c^H)^{-1}H_T\) is also at most \(|S|\). 

\subsection{Proof of Proposition~\ref{thm:random_layout}}
\label{app:proof_random_layout}

For uniform \(\Omega\), orthonormality implies \(\mathbb E[y_\ell(\Omega)y_\ell(\Omega)^{\mathsf H}]=(1/4\pi)I_{2\ell+1}\). The summands \(X_i=y_\ell(\Omega_i)y_\ell(\Omega_i)^{\mathsf H}\) are positive semidefinite and satisfy \(\lambda_{\max}(X_i)=(2\ell+1)/(4\pi)\), and \(\mathbb E[\Gamma_\ell]=(N_{\rm sens}/4\pi)I\). Matrix Chernoff~\cite{tropp2012tail} gives \eqref{eq:random_layout_chernoff}; Corollary~\ref{cor:separable_wave_law} gives the corresponding margin.

\section{Proof of Proposition~\ref{thm:operational_metrics}}
\label{app:proof_operational_metrics}

If \(\mu_\ell(S)\ge\tau\), all compressed degree-Gram eigenvalues are at least \(\tau\), which yields the log-det, posterior-variance, and EVM bounds. For \(\Delta s\in V_\ell\), the distance satisfies \(d^2(S)=\Delta s^{\mathsf H}G(S)\Delta s\ge\tau\|\Delta s\|_2^2\). Substituting this inequality into \(P(i\rightarrow j)\le \frac12\exp(-d^2(S)/4)\) yields the ML bound.

Applying the same argument to \(G_{TT}(S)\) proves the target-only bounds. Applying it to \(I_{T|c}(S)\), together with Theorem~\ref{thm:multitier-admission}, gives \(\log_2\det(I+I_{T|c}(S))\ge d_T\log_2(1+\tau)\) and \(\Sigma_{T|y}(S)\preceq(1+\tau)^{-1}I_{d_T}\).

\bibliographystyle{IEEEtran}
\bibliography{references}

\vfill
\end{document}